\documentclass[twoside,leqno,twocolumn]{article}  
\usepackage{ltexpprt} 

\usepackage{url}
\usepackage{amsmath,amssymb}
\usepackage{units}
\usepackage{color}

\newtheorem{assumption}{Assumption}[section]

\newtheorem{mythm}{Theorem}

\usepackage{pgf,tikz}
\usetikzlibrary{arrows}

\usepackage{marvosym}
\usepackage{graphicx}

\usepackage{ amssymb }

\newcommand{\dist} {\textrm{dist}}
\newcommand{\comment}[1]{}
\newcommand{\abs}[1]{|#1|}
\newcommand{\set}[1]{\{#1\}}
\newcommand{\PP}{{\mathcal{P}}}
\newcommand{\QQ}{{\mathcal{Q}}}
\newcommand{\lgO}{{\tilde{O}}}

\newcommand{\qed}{\hfill$\Box$}

\def\pagebreak{} 

\def\spherecut{n^{1+\epsilon}}

\begin{document}

\title{Exact Distance Oracles for Planar Graphs\thanks{Work by SM partially supported by NSF Grant CCF-0964037 and by
  a Kanellakis fellowship. Work by CS partially supported by the Swiss National Science Foundation and MADALGO, 
a Center of the Danish National Research Foundation.}}
\author{Shay Mozes\thanks{Brown University. E-mail: {\tt shay@cs.brown.edu}}
\and Christian Sommer\thanks{MIT. E-mail: {\tt csom@csail.mit.edu}}
}

\date{} 

\maketitle
\begin{abstract} \small\baselineskip=9pt 
We present new and improved data structures that answer exact node-to-node  distance queries in planar graphs. 
Such data structures are also known as {\em distance oracles}. 
For any directed planar graph on $n$ nodes with non-negative lengths we obtain the following:\footnote{Asymptotic notation as in $\lgO(\cdot)$ suppresses polylogarithmic factors in the number of nodes~$n$.}
\begin{itemize}
\item Given a desired space allocation $S\in[n\lg\lg n,n^2]$, we show how to construct in $\tilde O(S)$ time
 a data structure of size $O(S)$ that 
answers distance queries in $\lgO(n/\sqrt S)$ time per query. 
The best distance oracles for planar graphs until the current work are
due to Cabello (SODA 2006), Chen and Xu (STOC 2000), Djidjev (WG 1996), and
Fakcharoenphol and Rao (FOCS 2001). For $\sigma\in(1,4/3)$ and
space $S=n^\sigma$, we essentially improve the query time from $n^2/S$
to $\sqrt{n^2/S}$.
\item As a consequence, we obtain an improvement over the fastest algorithm for $k$--many distances 
in planar graphs whenever $k\in[\sqrt n,n)$. 
\item We provide a linear-space exact distance oracle
for planar graphs with query time $O(n^{1/2+\epsilon})$ for any constant $\epsilon>0$. This is the first such data structure with 
provable sublinear query time.
\item For edge lengths $\geq1$, we provide an exact distance oracle of 
space $\lgO(n)$ such that for any pair of nodes at distance~$\ell$ 
the query time is $\lgO(\min\set{\ell,\sqrt{n}})$. 
Comparable query performance had been observed experimentally but could not be explained theoretically. 
\end{itemize}
Our data structures with superlinear 
space are based on the following new tool: given a non-self-crossing 
cycle $C$ with $c = O(\sqrt n)$ nodes, we can preprocess $G$ 
in $\lgO(n)$ time to produce a data structure of size $O(n \lg\lg c)$
that can answer the following queries in $\lgO(c)$ time: 
for a query node $u$, output the distance from $u$ to all the nodes of~$C$. 
This data structure builds on and provides an alternative for a related data structure of
Klein (SODA 2005), which reports distances to the boundary of
a face, rather than a cycle.

\end{abstract}



\section{Introduction.}

A fast shortest-path query data structure may be of use whenever an application needs to compute 
shortest-path distances between some but not necessarily all pairs of nodes. 
Indeed, shortest-path query processing is an integral part of many applications~\cite{SommerThesis}, in particular in 
Geographic Information Systems (GIS), intelligent transportation systems~\cite{conf/cikm/JingHR96}. 
These systems may help individuals in finding fast routes or they may also assist companies in improving 
fleet management, plant and facility layout, and supply chain management. 
A particular challenge for traffic information systems or public transportation systems is to process a vast number of 
queries on-line while keeping the space requirements as small as possible~\cite{reference/algo/Zaroliagis}. 
Low space consumption is obviously very important when a query algorithm is run on a system with heavily
restricted memory such as a handheld device~\cite{conf/alenex/GoldbergW05} but it is also important for 
systems with memory hierarchies~\cite{journals/dam/HutchinsonMZ03,conf/isaac/ArgeT05}, where caching 
effects can have a significant impact on the query time. 

While many road  networks are actually not exactly 
planar~\cite{conf/gis/EppsteinG08,HighwayDimension}, they still share 
many properties with planar graphs; in particular, many road networks appear to have small separators as well. 
For this reason, planar graphs are often used to model various transportation networks. 

In the following, we provide shortest-path query data structures (also known as {\em distance oracles}) for 
planar graphs for essentially {\em any} specified space requirement.
Throughout the paper we assume the edge lengths to be non-negative.\footnote{Our results apply to graphs with negative-length edges by using reduced lengths induced by a feasible price function~\cite{Johnson77}. The current best bound for computing a feasible price-function in a planar graph is $O(n \lg^2n / \lg\lg n)$~\cite{esa/MozesW10}.}
Our results extend results of Cabello~\cite{conf/soda/Cabello06}, Chen and Xu~\cite{stoc/ChenX00}, and Fakcharoenphol and Rao~\cite{journals/jcss/FakcharoenpholR06} 
and extend one result and improve upon another result of Djidjev~\cite{DjidjevWG96}. 
Consider Figure~\ref{fig:tradeoff} for an illustration and comparison. 

\def\tradeoffthm{Let $G$ be a directed planar graph on $n$ vertices. For any value $S$ in the range 
$S\in[n\lg\lg n,n^2]$, there is a data structure with preprocessing time  $O(S\lg^3n/\lg\lg n)$ 
and space $O(S)$ that answers distance queries in $O(nS^{-1/2}\lg^{2}n\lg^{3/2}\lg n)$ time per query.}

\begin{theorem}
\tradeoffthm
\label{thm:general}
\end{theorem}

\newpage 
As a corollary, we obtain the following result on $k$--many distances~\cite{conf/wg/FeuersteinM91}, 
improving upon an algorithm of Cabello~\cite{conf/soda/Cabello06}, 
which runs in time $\lgO((kn)^{2/3} + n^{4/3})$. 
Our result is an improvement for $k=\tilde o(n)$.
For the range roughly $k\in[\sqrt n,n)$, our algorithm is faster by a factor of $\lgO((n/k)^{2/3})$.

\begin{theorem}
Let $G$ be a directed planar graph on $n$ vertices.
The distances between $k=\Omega(n^{1/2}\lg n/\lg\lg n)$ pairs of nodes $(s_1,t_1),(s_2,t_2),\dots (s_k,t_k)$ 
can be computed in time $O((kn)^{2/3}(\lg n)^{7/3}(\lg\lg n)^{2/3})$. 
\label{thm:kmanydist}
\end{theorem}

We also give a data structure that keeps the space 
requirements as small as possible, i.e.~linear in the size of the
input. This is  the first linear-space data 
structure with provably {\em sublinear} query time for exact point-to-point
shortest-path queries. Nussbaum~\cite{Nussbaum11} has simultaneously obtained a similar result.

\def\linearspacethm{For any directed planar graph $G$ with non-negative arc lengths and for any constant $\epsilon>0$, 
there is a data structure that supports 
exact distance queries in $G$ with the following properties: 
the data structure can be created in time $O(n\lg n)$, 
the space required is $O(n)$, and 
the query time is $O(n^{1/2+\epsilon})$.}

\begin{theorem}
\linearspacethm
\label{thm:linearspace}
\end{theorem}

More generally, our data structure works for the range $S\in[n,n\lg\lg n]$ (using a non-constant $\epsilon$, 
see details in Section~\ref{sec:linearspace}), where we slightly improve upon the tradeoff obtained by 
the construction of Fakcharoenphol and Rao~\cite{journals/jcss/FakcharoenpholR06}. 
Combined with Theorem~\ref{thm:general}, we provide efficient distance oracles for {\em any} space bound~$S\in[n,n^2]$. 

The main techniques we use are Frederickson's $r$--division~\cite{journals/siamcomp/Frederickson87}, 
Fakcharoenphol and Rao's implementation of Dijkstra's algorithm~\cite{journals/jcss/FakcharoenpholR06}, 
and Klein's Multiple-Source Shortest Paths (MSSP) data structure~\cite{conf/soda/Klein05}, for which we propose 
a more space-efficient alternative (see Section~\ref{sec:cyclemssp} for a detailed comparison) as follows.

\def\cyclemssp{Given a directed planar graph $G$
on $n$ nodes and a simple cycle $C$ with $c = O(\sqrt n)$ nodes, there is an algorithm that preprocesses $G$
in $O(n \lg^3 n)$ time to produce a data structure of size $O(n \lg\lg c)$
that can answer the following queries in $O(c \lg^2c \lg\lg c)$ time:
for a query node $u$, output the distance from $u$ to all the nodes of $C$.}

\begin{theorem}
\cyclemssp
\label{thm:cyclemssp}
\end{theorem}

Since Klein's MSSP data structure has found numerous applications, 
we believe that our data structure could be a rather useful tool in other algorithms as 
well.\footnote{There is at least one application for approximate distance oracles~\cite{journals/corr/abs-1109-2641}.} 

The shortest-path query problem has been investigated heavily from an experimental perspective as well. 
Experimental results suggest that query times proportional to the shortest-path length are possible in practice 
using algorithms based on the so-called {\em arc-flag} technique~\cite{lauther04,conf/wea/KohlerMS05,ArcflagsDimacs}.\footnote{
The preprocessing algorithm of this technique first partitions the graph into regions $V_i$ and thereafter labels each edge $e$ for all regions $V_i$ with a boolean {\em flag} $sp_i(e)$ indicating whether $e$ lies on any shortest path to any node in $V_i$. At query time, only edges leading towards the target region need to be considered. } 
Hilger, K\"ohler, M\"ohring, and Schilling make a statement about the (experimental) worst-case behavior of their method: 
\begin{quote}
In all cases, the search space of our arc-flag method is never larger
than ten times the actual number of nodes on the shortest paths~\cite[Section~6]{ArcflagsDimacs}. 
\end{quote}
We can now actually {\em prove} a similar statement: 
\begin{quote}
In (provably) all cases, the search space of our method is never larger
than a polylogarithmic factor times the length $\ell$ of the shortest paths.
\end{quote}

\def\pathlengthquerytime{For any planar graph $G$ with edge lengths $\geq1$ there is an exact distance oracle using space $O(n\lg n\lg\lg n)$ with 
query time $O(\min\left\{\ell\lg^2\ell\lg\lg\ell,\sqrt{n}\lg^2n\right\})$ for any pair of nodes at distance $\ell$. 
The preprocessing time is bounded by~$O(\spherecut)$ for any constant $\epsilon>0$.}

More precisely:
\begin{theorem}
\pathlengthquerytime
\label{thm:pathlengthquerytime}
\end{theorem}

Note that our data structure 
has query time proportional to the path {\em length}. In fact, our algorithm maintains a {\em Bellman-Ford-type} invariant: after iteration $i$, 
the distance represents the minimum path length among all paths on $\lgO(i)$ edges --- the correct distance is computed after time roughly
 proportional to the minimum number of edges on a shortest path but we can only guarantee correctness after time $\lgO(\ell)$. 
If we may further assume that, for some constant $\bar\epsilon>0$, all $s--t$ paths of length at most $(1+\bar\epsilon)d_G(s,t)$ have $\Omega(h_G(s,t))$ edges (where $h_G(s,t)$ denotes the number of edges ({\em hops}) on a minimum-hop shortest-path), 
 then our data structure can be constructed to have query time 
proportional to the minimum number of edges on a shortest path $\lgO(h_G(s,t))$.
This assumption essentially means that any $s-t$ path with significantly 
fewer edges than the shortest path is much longer. 

Our main contributions can be summarized as follows: 
{\em (i)} We improve the worst-case behavior of previously known distance oracles with low space requirements, 
providing a distance oracle with fast preprocessing algorithm that works for the whole tradeoff curve (we also provide the 
first one with linear space and sublinear query time), and 
{\em (ii)}  we make an important step towards proving the behavior observed in practice. As our main tool, 
{\em (iii)} we provide a more space-efficient multiple-source shortest-path data structure.

\subsection{Related Work.}
Shortest-path query processing for planar graphs has been studied extensively. 
In this section, we give a brief review of previous results. 
We focus on the space--query time tradeoff. 
See Figure~\ref{fig:tradeoff} for a summary of known results in comparison with ours. 
The tradeoffs previously had not been illustrated by a space vs.~query time plot as in Figure~\ref{fig:tradeoff}; indeed, 
the illustration suggests that a data structure like the one described in our main theorem (Theorem~\ref{thm:general}) 
was bound to exist. 

\definecolor{uququq}{rgb}{0.25,0.25,0.25}
\definecolor{xdxdff}{rgb}{0.49,0.49,1}
\begin{figure}[h!]
\begin{center}
\begin{tikzpicture}[line cap=round,line join=round,>=triangle 45,x=7cm,y=5cm]
\draw[->,color=black] (-0.1,0) -- (1.1,0);
\draw[color=black] (0.333333333333333,2pt) -- (0.333333333333333,-2pt) node[below] {\footnotesize $4/3$};
\draw[color=black] (1,2pt) -- (1,-2pt) node[below] {\footnotesize $2$};
\draw[color=black] (0.5,2pt) -- (0.5,-2pt) node[below] {\footnotesize $3/2$};
\draw[color=black] (2pt,0.25) -- (-2pt,0.25) node[left] {\footnotesize $1/4$};
\draw[color=black] (2pt,0.333333333333333) -- (-2pt,0.333333333333333) node[left] {\footnotesize $1/3$};
\draw[color=black] (2pt,0.5) -- (-2pt,0.5) node[left] {\footnotesize $1/2$};
\draw[color=black] (2pt,1) -- (-2pt,1) node[left] {\footnotesize $1$};
\draw[->,color=black] (0,-0.1) -- (0,1.1);
\draw[color=black] (1,-0.15) node {$\lg S/\lg n$};
\draw[color=black] (0.12,1.05) node {$\lg Q/\lg n$};
\clip(-0.1,-0.1) rectangle (1.1,1.1);
\draw (0,1)-- (1,0);
\draw[dashed] (0,0.5)-- (0.333333333333333,0.3333333333333);
\draw (0.333333333,0.33333333333)-- (0.5,0.25);
\draw (0.5,0.25)-- (1,0);
\fill [color=black] (0,0.5) circle (3pt);
\draw[color=black] (0.075,0.56) node {\cite{journals/jcss/FakcharoenpholR06}};
\draw[color=black] (0.45,0.82) node {\cite{DjidjevWG96}};
\fill [color=uququq] (0.5,0.5) circle (1.5pt);
\draw[color=black] (0.77,0.36) node {\cite{ArikatiCCDSZ96}};
\fill [color=uququq] (0.5,0.25) circle (1.5pt);
\fill [color=uququq] (0.33333333333333,0.333333333333333) circle (1.5pt);
\draw[color=black] (0.4,0.23) node {\cite{DjidjevWG96}};
\draw[color=black] (0.6,0.1) node {\cite{stoc/ChenX00,conf/soda/Cabello06}};
\draw[color=black] (0.15,0.37) node {Thm~\ref{thm:general}};
\end{tikzpicture}
\end{center}
\caption{Tradeoff of the Space [$S$] vs.~the Query time [$Q$] for different shortest-path query data structures on a doubly logarithmic scale, ignoring constant and logarithmic factors. The upper line represents the $Q=n^2/S$ tradeoff (completely covered by oracles of Djidjev~\cite{DjidjevWG96}; the oracles of Arikati, Chen, Chew, Das, Smid, and Zaroliagis~\cite{ArikatiCCDSZ96} cover the range $S\in[n^{3/2},n^2]$; SSSP ($S=n$) and APSP ($S=n^2$) also lie on this line). The lower line represents the $Q=n/\sqrt S$ tradeoff; the result of Djidjev~\cite{DjidjevWG96} covers the range $S\in[n^{4/3},n^{3/2}]$; Chen and Xu~\cite{stoc/ChenX00} and Cabello~\cite{conf/soda/Cabello06} extend this to $S\in[n^{4/3},n^2]$. The data structure of Fakcharoenphol and Rao~\cite{journals/jcss/FakcharoenpholR06} covers the point $S=n$ and $Q=\sqrt n$. We extend their results to the full range.\label{fig:tradeoff}}
\end{figure}
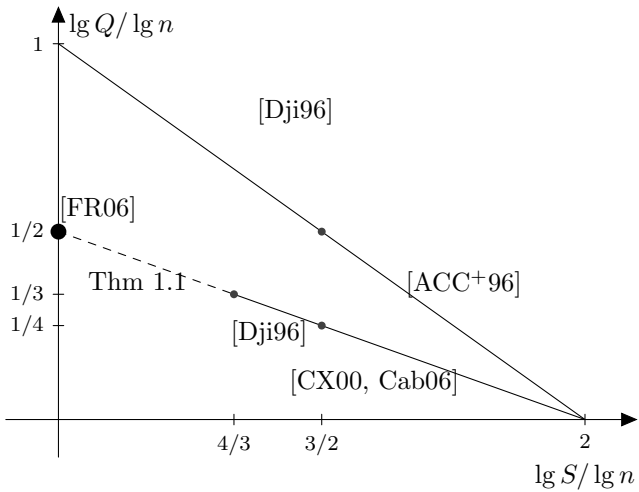

For exact shortest-path queries, the currently best result in terms of the tradeoff between space and query time is by Fakcharoenphol and Rao~\cite{journals/jcss/FakcharoenpholR06}. 
Their data structure of space $\lgO(n)$ can be constructed in time $\lgO(n)$ and processes queries in time $\lgO(\sqrt n)$. 
The preprocessing time can be improved by a logarithmic factor~\cite{journals/talg/KleinMW10}. 
(Applications to distance oracles are not explicitly stated in~\cite{journals/talg/KleinMW10}.) 
%

Some distance oracles have better query times. 
Djidjev~\cite{DjidjevWG96}, improving upon earlier work of Feuerstein and Marchetti-Spaccamela~\cite{conf/wg/FeuersteinM91}, 
proves that, for any $S\in[n,n^2]$, there is an exact distance oracle with 
preprocessing time $O(S)$ (which increases to $O(n\sqrt S)$ for $S\in[n,n^{3/2}]$), space $O(S)$, and query time $O(n^2/S)$. 
For a smaller range, he also proves that, for any $S\in[n^{4/3},n^{3/2}]$, 
there is an exact distance oracle with preprocessing time $O(n\sqrt
S)$, space $O(S)$, and query time $\lgO(n/\sqrt{S})$. 
Chen and Xu~\cite{stoc/ChenX00}, extending the range, prove that, for any $S\in[n^{4/3},n^2]$, 
there is an exact distance oracle using space $O(S)$ with preprocessing time $O(n\sqrt S)$ and 
query time $\lgO(n/\sqrt S)$. 
Cabello~\cite{conf/soda/Cabello06}, mainly improving the preprocessing time, proves that, for any $S\in[n^{4/3},n^2]$, 
there is an exact distance oracle with preprocessing time and space $O(S)$ and query time $\lgO(n/\sqrt S)$. 
Compared to Djidjev's construction, the query time is slower by a logarithmic factor but the range for $S$ is larger and the preprocessing time is faster. 
Nussbaum~\cite{Nussbaum11} provides a data structure that maintains both the tradeoff from~\cite{DjidjevWG96,stoc/ChenX00} and the fast preprocessing time from~\cite{conf/soda/Cabello06}. 
Nussbaum also provides a different data structure with a ``clean'' tradeoff of space $O(S)$ and query time $O(n/\sqrt S)$, however spending time 
$\lgO(S^{3/2}/\sqrt n)$ in the preprocessing phase. 
In our construction, we sacrifice another root-logarithmic factor in the query time (compared to~\cite{conf/soda/Cabello06}) 
but we prove the bounds for essentially the whole range of~$S$; our preprocessing time remains~$\lgO(S)$.

If constant query time is desired,  storing a complete distance matrix is
essentially the best solution we have. Wulff-Nilsen~\cite{WNThesis} 
recently improved the space requirements to $o(n^2)$. 
If the space is restricted to linear, using the linear-time single-source shortest path algorithm of 
Henzinger, Klein, Rao, and Subramanian~\cite{journals/jcss/HenzingerKRS97} is the fastest known for exact shortest-path queries 
until the current work (Theorem~\ref{thm:linearspace}). 
Nussbaum~\cite{Nussbaum11} has simultaneously obtained a result similar to Theorem~\ref{thm:linearspace}.

Efficient data structures for shortest-path queries have also been devised for restricted classes 
of planar graphs~\cite{journals/algorithmica/DjidjevPZ00,stoc/ChenX00} and for restricted types
of queries~\cite{journals/jgaa/Eppstein99,journals/talg/KowalikK06,journals/siamcomp/Schmidt98,conf/soda/Klein05}. 
If approximate distances and shortest paths are sufficient, a better tradeoff with $\lgO(n)$ space and $\lgO(1)$ query time has been 
achieved~\cite{ThorupJACM04,conf/soda/Klein02,conf/soda/Klein05,KKS}.

Based on separators, geometric properties, and other characteristics such as highway structures, many efficient  practical methods have been devised~\cite{GeisbergerSSD08,SchultesSanders05,TransitNodesScience}, their time and space complexities are however difficult to analyze. 
Competitive worst-case bounds have been achieved 
under the assumption that actual road networks have small {\em highway dimension}~\cite{HighwayDimension,conf/icalp/AbrahamDFGW11}. 
While our preprocessing algorithm (Theorem~\ref{thm:pathlengthquerytime}) runs in almost linear time, 
some of the problems that appear in the preprocessing stage of practical route planning methods 
have recently been proven to be NP-hard~\cite{ShortcutProblem,conf/ciac/BauerCKKW10}.



\pagebreak

\section{Preliminaries.}
\label{sec:prel}
\subsection{Recursive $r$--division of Planar Graphs.}
\label{sec:rdiv}
Let $G = (V,E)$ be a planar graph with $\abs{V} = n$. Let $E_P$ be a subset of $E$, and let $P = (V_P,E_P)$ be the subgraph of $G$ induced by
$E_P$. $P$ is called a \emph{piece} of $G$. The nodes of $V_P$ that
are incident in $G$ to nodes of $V \setminus V_P$ are called the
{\em boundary nodes} of $P$ and denoted by $\partial P$.

An {\em $r$--division}~\cite{journals/siamcomp/Frederickson87} of $G$ is a decomposition into $O(n/r)$
edge-disjoint pieces,
each with $O(r)$ nodes and $O(\sqrt{r})$ boundary nodes. We use an
$r$--division with the additional property that, in each piece,
there exist a constant number of faces, called {\em holes}, such that every
boundary node is incident to some hole. Such a decomposition can be
found in $O(n \lg r + nr^{-1/2}\lg n)$~\cite{nlglgn-mincut-WN10} by
applying Miller's cycle separator~\cite{journals/jcss/Miller86} iteratively.

We use this $r$--division recursively. 
Let the base of the recursion be level $0$, and let the top of the recursion be level~$k$.
The entire graph $G$ is defined to be the only piece at level~$k$. 
The pieces of level $i$ of the recursion are obtained by computing an
$r_i$--division for each level--$(i+1)$ piece. The notation $r_i$
suggests that we may use a different parameter $r$ in the
$r$--division at every level of the recursion. Indeed, using 
a non-uniform recursion is important in obtaining Theorem~\ref{thm:linearspace}. 
For a level--$i$ piece $P$, the level--$(i-1)$ pieces obtained by applying
the $r$--division to $P$ are called the \emph{subpieces} of $P$.

We stress that the
classification of nodes of a piece at \emph{any} level as boundary
nodes is with respect to $G$ (and not $P$). This implies that, if $v$ is a
boundary node of a level--$i$ piece, then $v$ is also a boundary node of
any lower-level piece that contains $v$.
This generalizes the decomposition used in~\cite{journals/jcss/FakcharoenpholR06}, where Miller's separator is used at each
level rather than an $r$--division.

\subsection{Klein's Multiple-Source Shortest Paths Algorithm.}\label{subsec:mssp}
Klein~\cite{conf/soda/Klein05} gave a multiple-source shortest path (MSSP) algorithm with the following properties.  
The input consists of a directed
planar embedded graph $G$ with non-negative arc-lengths, and a face $f$. For each node $u$ in turn on the boundary of $f$, the
algorithm computes (an implicit representation of) the shortest-path
tree rooted at $u$.  
This takes a total of $O(n \lg n)$ time and space. 
Subsequently, the distance between any  pair $(u,v)$ of
nodes of $G$ where  $u$ is on the boundary of $f$ (and $v$ is an arbitrary node), can be queried in $O(\lg n)$ time.
In particular, given a set $S$ of $O(\sqrt{n})$ nodes on the boundary
of a single face, the algorithm can compute all $S$-to-$S$ distances
in $O(n \lg n)$ time. 

We propose a more space-efficient alternative in Section~\ref{sec:cyclemssp}, 
wherein we also provide a detailed comparison. 

\subsection{Dense Distance Graphs and Efficient Implementation of
  Dijkstra's Algorithm.}\label{subsec:ddgdijkstra}

The {\em dense distance graph} for a piece $P$, denoted by $DDG_P$, is the
complete graph on $\partial P$, the
boundary nodes of $P$, such that the length of an arc corresponds to
the distance (in $P$) between its endpoints. The dense distance graph
for all pieces $P$ in the $r$--division can be computed in $O(\abs{V(G)} \lg \abs{V(G)})$ time and
space using Klein's MSSP; For each piece $P$, compute MSSP data structures in
$O(\abs{P} \lg \abs{P})$ time and space a constant number of times,
specifying a different hole of $P$ as the distinguished face at each
run. Then query the MSSP data structures for the distances between
the boundary nodes in $O(\abs{\partial P}^2 \lg \abs{P})$ time. Since in
  an $r$--division $\abs{\partial P} = \sqrt{\abs{P}}$, this takes
  $O(\abs{P} \lg \abs{P})$ time and space per piece, for a total of
  $O(\abs{V(G)} \lg \abs{V(G)})$ for all pieces.

Let $\PP$ be a set of pieces (not necessarily at the same level), 
and let the graph $H$ be the union of the dense distance graphs of the pieces in $\PP$.
Fakcharoenphol and Rao~\cite{journals/jcss/FakcharoenpholR06} devised an 
ingenious implementation of Dijkstra's algorithm~\cite{Dijkstra59} 
that computes a shortest-path tree in $H$ in time $O(\abs{V(H)}\lg^2n)$, where
$\abs{V(H)}$ is the number of nodes in $H$ (i.e., the total number of boundary
 nodes in all the pieces in $\PP$). We will refer to this
 implementation as {\em FR-Dijkstra}.

In fact, the proof of Fakcharoenphol and Rao's algorithm only relies on the
property that the distances in each of the dense distance graphs given
as input correspond to distances in a planar graph between a set of
nodes that lie on a constant number of faces. It does not rely on any
other properties of the $r$--division.

FR-Dijkstra can be extended
to the following setting (cf.~\cite{conf/focs/BorradaileSW10}). Let $J$ be a planar graph. Let $n'$ denote
the number of nodes of $J \cup H$.
We can compute shortest paths in $H \cup J$ in $O(\abs{H} \lg^2
\abs{H} + n' \lg n')$ time. The edges of $H$ are relaxed using
FR-Dijkstra, while the edges of
$J$ are relaxed as in a traditional implementation of Dijkstra's
algorithm using a heap.

\subsubsection{External Dense Distance Graphs.}\label{sec:extddg}
For a piece $P$, let $G-P$ be the graph obtained from $G$ by deleting
the nodes in $V_P - \partial P$. 
The {\em external} dense distance graph of $P$,  denoted by $DDG_{G-P}$, is the
complete graph on $\partial P$ such that the length of an arc corresponds to
the distance between its endpoints in $G- P$. External dense
distance graphs were used recently in~\cite{conf/focs/BorradaileSW10}. Computing the
external dense distance graphs for all pieces in an $r$--division
cannot be done efficiently using Klein's MSSP. The reason is that
$\abs{V(G-P)}$ can be $\Theta(\abs{V(G)})$ for all pieces. Instead, the
computation can be done in a top-down approach as follows
(cf.~\cite{conf/focs/BorradaileSW10}).
Recall that an
  $r$--division is obtained as the set of pieces of the lowest level
  in a recursive division of the graph using Miller's simple-cycle
  separators. Consider the set $\QQ$ of all pieces at all levels of the
  recursion (rather than just the set of pieces at the lowest
  level). Note that there are $O(\lg \abs{V(G)})$ recursive levels since
  the size of the pieces decreases by a constant factor for every constant number of
  applications of Miller's cycle separator theorem.
First, we compute $DDG_Q$ for every piece $Q \in \QQ$. As explained above,
this can be done in $O(\abs{V(G)} \lg \abs{V(G)})$ for all the pieces in a
specific level of the recursion, for a total of $O(\abs{V(G)} \lg^2
  \abs{V(G)}  )$ for all pieces in~$\QQ$. 
Next, we consider a piece $Q$ and denote the two subpieces of $Q$ by $Q_1$ and
$Q_2$. $DDG_{G-{Q_2}}$ is obtained by computing distances 
in $DDG_{Q_1} \cup DDG_{G-Q}$ (see Figure~\ref{fig:extDDG}), using multiple
 applications of FR-Dijkstra, once for each node in $\partial Q_2$.
 This takes $O(\abs{Q} \lg^2
  |\partial Q|)$ time per piece, for an overall $O(\abs{V(G)}
  \lg^2\abs{V(G)} )$ time for
  all pieces in a specific level. Since the number of levels is
  bounded by $O(\lg \abs{V(G)})$, the entire computation takes
  $O(\abs{V(G)}
  \lg^3\abs{V(G)} )$ time.

\begin{figure}[h]
\centerline{\includegraphics[scale=0.4]{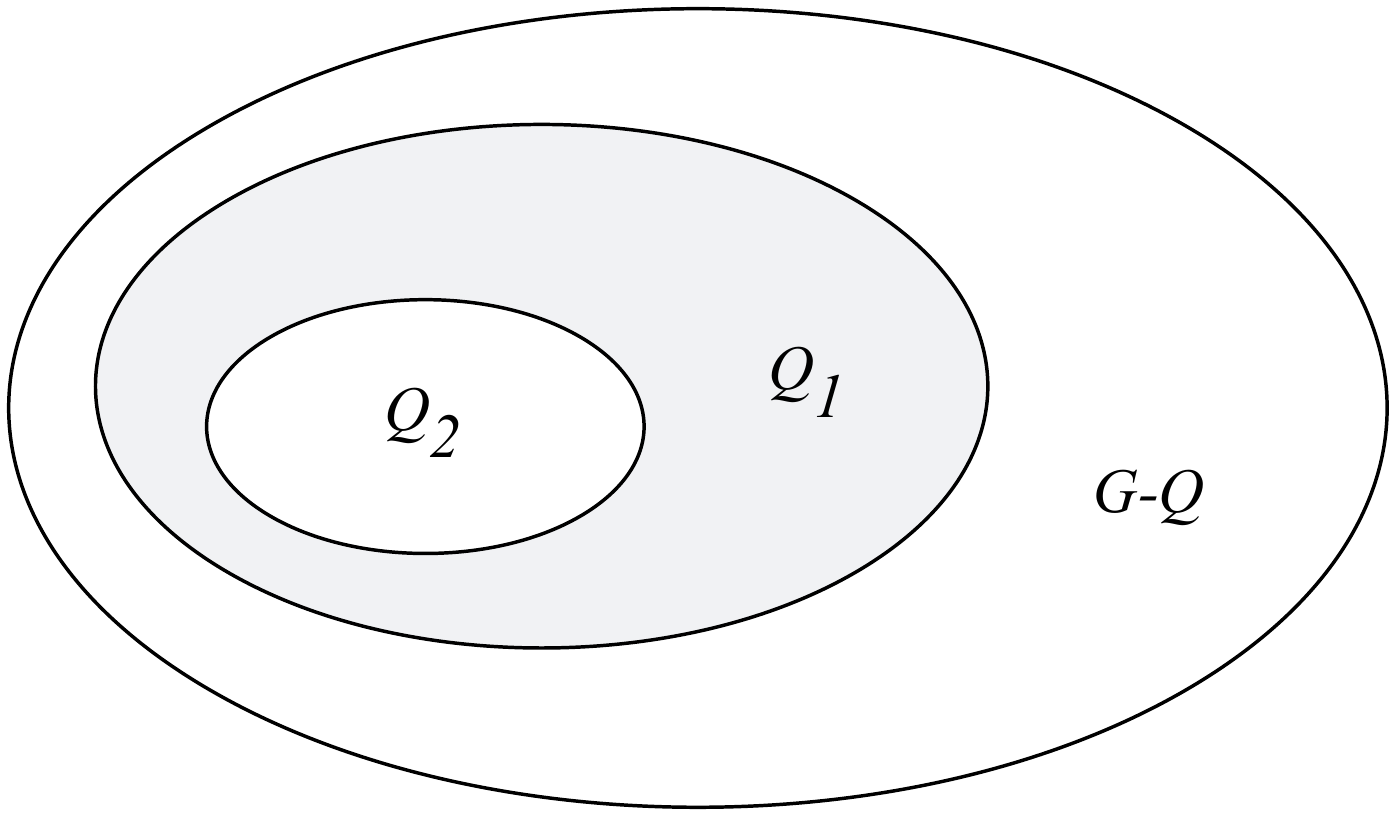}}
\caption{Pieces $Q_1$ and $Q_2$ in some level of the recursive application of Miller's cycle separator theorem. 
The piece $Q$ is the union of $Q_1$ and $Q_2$. 
Distances in $G-Q_2$, the exterior of $Q_2$, are obtained by considering shortest paths in the interior of $Q_1$ and in $G-Q$, the exterior of~$Q$.}
\label{fig:extDDG}
\end{figure}

\section{A Linear-Space Distance Oracle.}\label{sec:linearspace}

We first provide our linear-space data structure. The techniques
employed here are reused in the 
subsequent sections, particularly in the cycle MSSP data structure (Section~\ref{sec:cyclemssp}). 
In the following we prove a more precise version of Theorem~\ref{thm:linearspace}.

{
\renewcommand{\themythm}{\ref{thm:linearspace}}
\begin{mythm}
\linearspacethm

For non-constant $\epsilon>0$, the 
preprocessing time is $O(n\lg (n) \lg({1/\epsilon}))$, 
the space required is $O(n\lg({1/\epsilon}))$, and 
the query time is $O(n^{1/2+\epsilon}+n^{1/2}\lg^2(n)\lg ({1/\epsilon}))$. 
\end{mythm}
}

Our distance oracle is an extension of the oracle in Fakcharoenphol and Rao~\cite{journals/jcss/FakcharoenpholR06}. 
The main ingredients of our improved space vs.~query time tradeoff are {\em (i)} using recursive $r$--divisions instead of cycle separators, 
 and {\em (ii)} using an {\em
  adaptive} recursion,\footnote{This appears to be the main difference to the distance oracle of Nussbaum~\cite{Nussbaum11}, which uses a non-adaptive recursion. Using our adaptive recursion is crucial whenever $S\in[n\lg\lg n,n\lg n)$.}  where the ratio between the boundary sizes of 
piece at consecutive levels is $\sqrt n$.

We split the proof into descriptions and analysis of the preprocessing and query algorithms. 
Let $k =  \Theta(\lg(1/\epsilon))$.

\paragraph{Preprocessing.} We compute the
recursive $r$--division of the graph with $k$
recursive levels and values of $\{r_i\}_{i=0}^k$  to be
specified below. This takes $O(kn\lg n)$ time.
We then compute the dense distance graph for each
piece. This is done for a piece $P$, with $r$ nodes and $O(\sqrt{r})$ boundary nodes
on a constant number of holes, by applying Klein's MSSP
algorithm~\cite{conf/soda/Klein05} as described in Section~\ref{subsec:ddgdijkstra}.
Thus, all of the boundary-to-boundary
distances in $P$ are computed in $O(r \lg r)$ time. 
Summing over
all $O(n/r_i)$ pieces at level $i$, 
the preprocessing time per level is 
 $O(n \lg r_i)$. The overall time to compute 
the dense distance graphs for all
pieces over all recursive levels is therefore $O(kn \lg n)$. 

The space required to store $DDG_P$ is $O((\sqrt{r})^2) = O(r)$; summing over
all pieces at level $i$ we obtain space $O(\frac{n}{r_i}r_i) = O(n)$ per level; 
the total space requirement is $O(kn)$.

\paragraph{Query.}
Given a query for the distance between nodes $u$ and $v$, we proceed as
follows.
For simplicity of the presentation, we initially assume that neither $u$ nor $v$ are
boundary nodes.

Let $P_0$ be the level--$0$ piece that contains $u$.
We compute distances from $u$ in $P_0$. This is done in $O(r_0)$ time using the algorithm
of Henzinger et al.~\cite{journals/jcss/HenzingerKRS97}. Denote these
distances by $\dist_{P_0}(u,w)$ for $w\in P_0$. 
Let $H_0$ denote the star graph with center $u$ and leaves~$w\in\partial
P_0$.
The arcs of $H_0$ are directed from $u$ to the leaves, and their
lengths are the corresponding distances in $P_0$.

Let $S_u$ be the set of pieces that contain $u$. Note that $S_u$
contains exactly one piece of each level.
Let $R_u$ be the union of subpieces of every piece in $S_u$. That is,
$R_u = \bigcup_{P \in S_u} \{P': P' \textrm{ is a subpiece of } P \}$.
Let $H_u$ be the union of the dense distance graphs of the pieces in~$R_u$. 
We use FR-Dijkstra (see
Section~\ref{subsec:ddgdijkstra}) to compute distances from $u$ in
$H_u \cup H_0$.
Observe that any shortest path from $u$ to a node of $H_u$ can be decomposed into a
shortest path in $P_0$ from $u$ to $\partial P_0$ and shortest paths
each of which is between boundary nodes of
some piece in $R_u$. 
Since all $u$-to-$\partial P_0$ shortest paths  in $P_0$ are represented
in $H_0$, 
and since all shortest paths between boundary nodes of pieces in $R_u$
 are represented in $H_u$, this observation 
implies that distances from $u$ to nodes of $H_u$ in $H_u \cup H_0$
are equal to distances from $u$ to nodes of $H_u$ in $G$. We denote
these distances by $\dist_G(u,w)$ for nodes $w\in H_u$.
 
We repeat a similar procedure for $v$ (reversing the direction of arcs)
 to compute $\dist_G(w,v)$, the distances in $G$ from every node $w\in H_v$ to~$v$. 

Let $P_{uv}$ be the lowest-level piece that contains both $u$ and
$v$. Assume first that $P_{uv}$ is not a level--$0$ piece. 
Let $P_u$~($P_v$) be the subpiece of $P_{uv}$ that contains $u$~($v$). 
Since $P_{uv}$ is both in $S_u$ and in $S_v$, both $P_u$ and
$P_v$ are in $R_u$ as well as in $R_v$. This implies that we have
already computed $\dist_G(u,w)$ and $\dist_G(w,v)$ for all $w \in \partial P_u$.
Since we have assumed that $P_{uv}$ is not a level--$0$ piece, the 
shortest $u$-to-$v$ path must contain some node of $\partial P_u$.
Therefore, the $u$-to-$v$ distance can be found by computing  
$$\min_{w \in \partial P_u} \dist_G(u,w) + \dist_G(w,v).$$

If $P_{uv}$ is a level--$0$ piece, then $P_{uv} = P_0$ , and the $u$-to-$v$ distance can be
found by computing 
$$\min \left\{ \dist_{P_0}(u,v) , \min_{w \in \partial P_{uv}}
    \{ \dist_G(u,w) + \dist_G(w,v) \} \right\}.$$ 

The case when $u$ or $v$ are boundary nodes is a degenerate case that 
can be solved by the above algorithm. Let $Q_u$ be the highest-level piece of which $u$
is a boundary node. We have the preprocessed distances in $Q_u$ from $u$ to all
other nodes of $\partial{Q_u}$. Therefore, it suffices to replace
$S_u$ above with the set of pieces that contain $Q_u$ as a subgraph in order to assure that
$H_u$ is small enough and that 
the distances computed by the fast implementation of Dijkstra's
algorithm are the distances from $u$ to nodes of $H_u$ in $G$.

\paragraph{Query Time.}  Computing the
distances $\dist_{P_0}(\cdot,\cdot)$ takes $O(r_0)$ time. 
Let $\abs{V(H_u)}$ denote the number of nodes
of $H_u$. The FR-Dijkstra implementation runs in 
$O(\abs{V(H_u)} \lg^2 n)$ time. It therefore remains to bound $\abs{V(H_u)}$. 
Let $P_i$ be the level--$i$ piece in $S_u$. $P_i$~has $O(\frac{r_i}{r_{i-1}})$
subpieces, each with $O(\sqrt{r_{i-1}})$ boundary nodes. Therefore, the
contribution of $P_i$ to $\abs{V(H_u)}$ is $O(\frac{r_i}{\sqrt{r_{i-1}}})$. 
The total running time is therefore 
$$ O\left(r_0 + \lg^2n \sum_{i=1}^k \frac{r_i}{\sqrt{r_{i-1}}}\right).$$
Recall that $r_k = \abs{V(G)} = n$, and set $r_0 = \sqrt{n}$.
For $i=1 \dots k-1$ we recursively define $r_i$ so as to satisfy
$$\frac{r_i}{\sqrt{r_{i-1}}} = \sqrt{n} \nonumber.$$
This implies 
\begin{eqnarray*}
r_1 &=&  n^{\frac{1}{2} +  \frac{1}{4}}   =
n^{1 - \frac{1}{4}}  \\
r_2 &=&  n^{\frac{1}{2} +
  \frac{3}{8}}  = n^{1 -\frac{1}{8}}  \\
r_3 &=&  n^{\frac{1}{2} +
  \frac{7}{16}} = n^{1 - \frac{1}{16}} \\
& \hdots & \\
r_{k-1} & = & n^{1 - \frac{1}{2^k}}.
\end{eqnarray*}
The total running time is thus bounded by
\begin{eqnarray} \label{eq:Hu}
O\left(\sqrt{n} + \lg^2n\left( (k-1)\sqrt{n} + \frac{n}{\sqrt{n^{1 -
        \frac{1}{2^k}}}}\right)\right) & \leq & \nonumber \\
O\left( \left( k\sqrt{n} + n^{\frac{1}{2}
    + \frac{1}{2^{k+1}}} \right) \lg^2n\right). & &
\end{eqnarray} 
By setting $k =  \Theta(\lg(1/\epsilon))$ we obtain the claimed running times.

\pagebreak

\section{A Cycle MSSP Data Structure for Planar Graphs.}
\label{sec:cyclemssp}
In this section we provide our main technical tool. We prove Theorem~\ref{thm:cyclemssp}, which we restate here. 

{
\renewcommand{\themythm}{\ref{thm:cyclemssp}}
\begin{mythm}
\cyclemssp
\end{mythm}
}


\paragraph{Comparison with Klein's MSSP data structure}
\label{sec:cyclemssp:comparison}
Our data structure can be seen as an alternative to Klein's MSSP data structure (see Section~\ref{subsec:mssp}) 
with two main advantages  (which we exploit in Section~\ref{sec:general:wholerange}): 
\begin{itemize}
\item our data structure can handle queries to any not-too-long cycle as opposed to a single face (the crucial difference and difficulty is that shortest paths may cross a cycle but not a face),
\item the space requirements are only $O(n\lg\lg c)$ (we internally rely on the data structure in Theorem~\ref{thm:linearspace}, so even $O(n)$ is possible at the cost of increasing the query time) 
as opposed to $O(n\lg n)$, 
\end{itemize}
and three main disadvantages: 
\begin{itemize}
\item our data structure cannot efficiently answer queries from $u$ to a {\em single} node 
on the cycle $C$; such a query requires the same time as computing the distances from $u$ to all the nodes on $C$
(in many existing applications, this disadvantage is not really problematic, since MSSP is used for all nodes on the face anyway),
\item our data structure requires amortized time $O(\lg^2c\lg\lg c)$ per node on the cycle (as opposed to $O(\lg n)$), 
which is slower for long cycles, and 
\item the preprocessing time of our data structure is $O(n\lg^3n)$ as opposed to $O(n\lg n)$. 
\end{itemize}

Let $G$ be an embedded planar graph. 

\paragraph{Preprocessing.}
Let $G_0$ be the exterior of~$C$. That is, the graph obtained from $G$ by deleting all
nodes strictly enclosed by $C$. Consider $C$ as the infinite face of $G_0$.
Similarly, let $G_1$ be the interior of $G$. Namely, the graph obtained from $G$ by deleting all
nodes not enclosed by $C$. Consider $C$ as the infinite face of $G_1$. Note that we have 
reduced the problem from query nodes on a cycle $C$ to query nodes on a face. The query algorithm 
is supposed to handle paths that cross $C$. 
The preprocessing step consists of the following:
\begin{enumerate}

\item Computing $DDG_C$ and $DDG_{G- C}$. This can
be done in $O((n + c^2) \lg n)$ time using Klein's MSSP
algorithm~\cite{conf/soda/Klein05}.  Storing $DDG_C$ and $DDG_{G- C}$
requires $O(c^2) = O(n)$ space.

\item Computing an $r$--division of $G_i$ ($i\in\{0,1\}$) with $r=c^2$. Each
piece has $O(c^2)$ nodes and $O(c)$ boundary nodes incident to a
constant number of holes. Consider \emph{all} the nodes of $C$ as boundary nodes
of every piece in the division. 
Note that each piece still has $O(c)$ boundary
nodes. 
This step takes $O(n \lg n)$ time.
 
\item \label{Precursicve} Computing, for each piece $P$,  a recursive $r$--division of $P$ as the one in the preprocessing step of the oracle in
  Section~\ref{sec:linearspace} (Theorem~\ref{thm:linearspace}) with $\epsilon = 1/\lg c $. That is, the
  number of levels in this recursive $r$--division is $k = \Theta(\lg \lg
  c)$. The top-level (level--$k$) piece in this recursive division is the entire piece $P$. In the
  description in Section~\ref{sec:linearspace}, the top-level piece
  is the entire graph and therefore it has no boundary nodes. Here, in
  contrast, we consider the boundary nodes of $P$ as boundary nodes of
  the top-level piece in the decomposition (and thus, as boundary
  nodes of any lower-lever piece in which they appear). This does not asymptotically
  change the total number of boundary nodes at any level of the
  recursive decomposition since $P$ has $O(c)$ boundary
  nodes, and every level of the recursive decomposition consists of a
  total of $\Omega(c)$ boundary nodes. 
The time to compute the recursive $r$--division for all pieces is bounded by $O(n \lg^2 n)$.

\item \label{Poracledist} Computing, for each piece $P$, the dense distance
  graph for each of the pieces
  in the recursive decomposition of~$P$. Let $H_P$ denote the union of the dense
  distance graphs for all the pieces in the recursive decomposition
  of $P$. As discussed in Section~\ref{sec:linearspace}, the space
  required to store $H_P$ is $O(\abs{P} \lg \epsilon) = O(\abs{P} \lg\lg c)$. Using the methods presented
  in Section~\ref{sec:linearspace}, computing $H_P$ takes $O(\abs{P}\lg \abs{P} \lg
  \lg c) = O(c^2 \lg c \lg\lg c)$. Thus, the total space and time required over all pieces $P$ is $O(n
\lg \lg c)$ and  $O(n \lg c \lg\lg c)$, respectively.
\item \label{Pexternal} Computing, for each piece $P$, the dense distance graph $DDG_{G_i-P}$. Recall
  that we 
  consider the nodes of $C$ as boundary nodes of {\em every} piece
  in the division. These dense distance graphs can be computed as
  described in Section~\ref{sec:extddg}. 
  As shown there, the entire computation (for all pieces combined) takes
  $O(n \lg^3n)$ time. 
\end{enumerate}
The time required for the preprocessing step is therefore $O(n \lg^3
n)$ and the space required is $O(n \lg\lg c)$.
\begin{figure*}[ht!]
\centerline{\includegraphics[scale=0.6666]{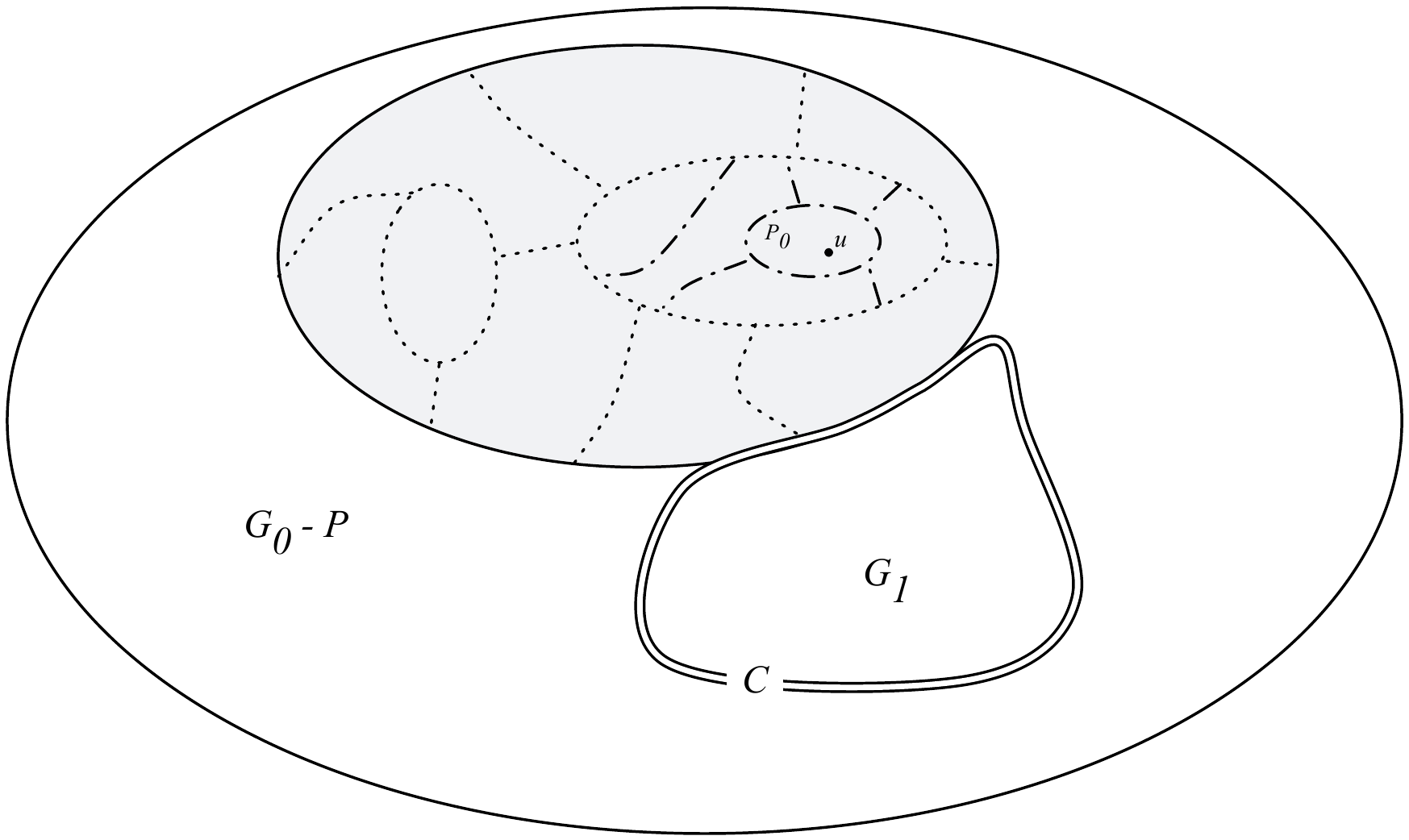}}
\caption{A schematic diagram showing the various subgraphs whose dense
distance graphs are used in a query to the cycle MSSP data structure. The
cycle $C$ is double-lined. The interior of $C$
is the subgraph $G_1$. The query node $u$ is indicated by a small solid circle. 
The piece $P$ in the
$r$--division of the exterior of $C$ ($G_0$) is shown as a 
grey region with solid boundary. The boundaries of the pieces whose
dense distance graphs are in $H_u$ are shown as dotted lines (one level) and dashed-dotted lines (another level). $P_0$ is the smallest
piece of $P$ that contains $u$. Any shortest path
from $u$ to $C$ can be decomposed into a shortest path from $u$ to
$\partial P_0$ followed by shortest paths between nodes on the boundaries
shown in the figure.}
\label{fig:cycleDS}
\end{figure*}

\paragraph{Query.} When queried with a node~$u$, the data structure
outputs the distances from $u$ to all the nodes of~$C$. We describe the
case where $u$ is not enclosed by~$C$. In this case, we use the dense
distance graphs computed in the preprocessing step for~$G_0$. 
The symmetric case is handled
similarly, by using the dense distance graphs computed for~$G_1$.

Let $P$ be the piece in the $r$--division of $G_0$ to which $u$
belongs. Recall that $P$ consists of $O(c^2)$ nodes.
Consider the recursive $r$--division of $P$ computed in
item~\ref{Precursicve} of the preprocessing stage. 
Let $P_0$ be the level--$0$ piece of $P$ that contains $u$. $P_0$
consists of $O(\sqrt{c^2}) = O(c)$ nodes.\footnote{Here, as in 
  Section~\ref{sec:linearspace}, we assume that $u$ is not a boundary
  node in the recursive $r$--division of $P$. The case where $u$
  is such a boundary node is degenerate, see Section~\ref{sec:linearspace}.}   

We first compute~\cite{journals/jcss/HenzingerKRS97}, in $O(c)$
time, the distances from $u$ to all nodes of $P_0$, and store them in a
table $\dist_{P_0}$.
We then compute, using FR-Dijkstra, the distances from $u$ in the union of the following
dense distance graphs (see Figure~\ref{fig:cycleDS}):
\begin{enumerate}
\item $H_0$, the star graph with center $u$ and leaves $\partial
  P_0$. The arcs of $H_0$ are directed from $u$ to the leaves and
  their lengths are the corresponding lengths in  $\dist_{P_0}$.
\item $H_u$, the subset of dense distance graphs in $H_P$ that correspond to
pieces in the recursive decomposition of $P$ that contain $u$ and
their subpieces. These dense distance graphs are available in $H_u$.
\item $DDG_{G_0-P}$
\item $DDG_{C}$ 
\end{enumerate}
Note that the first two graphs are the analogs of $H_0$ and $H_u$ from
Section~\ref{sec:linearspace}.

Distances from $u$ in the union of the above graphs are equal to  the distances from $u$  
in $G$. This is true since  any $u$-to-$C$
shortest path can be decomposed into 
{\em (i)} a shortest path in $P_0$ from $u$ to
$\partial P_0$, 
{\em (ii)} shortest paths each of which is a shortest
path in $Q$ between boundary nodes of $Q$ for some piece $Q$ in the
recursive $r$--division of $P$ that is
represented in $H_u$, 
{\em (iii)} shortest paths in $G_0-P$ between nodes of $\partial P \cup C$,
and {\em (iv)} shortest paths in the interior of $C$ between nodes of~$C$. 

To bound the running time of FR-Dijkstra we need to bound the number of nodes in
all dense distance graphs used in the FR-Dijkstra computation. 
$H_0$ has $O(\sqrt c)$ nodes. The
analysis in Section~\ref{sec:linearspace} shows that the graphs
in the set $H_u$ consist of $O(\sqrt{\abs{P}} \lg\lg \abs{P})$ nodes
(substitute $k=1/\lg\abs{P}$ in eq.~(\ref{eq:Hu})). 
$DDG_{G_0-P}$ has $O(c + \sqrt{\abs P}) = O(c)$ nodes, and 
$DDG_{C}$ has $c$ nodes. Combined, 
the running time of the invocation of FR-Dijkstra is bounded
by $O(c \lg^2 c \lg \lg c)$. This dominates the $O(c)$ time required
for the computation  of $\dist_{P_0}$, so the overall query time is
$O(c \lg^2 c \lg \lg c)$, as claimed.

\section{Distance Oracles with Space $S\in[n\lg\lg n,n^2]$.}
\label{sec:general:wholerange}
\label{sec:general:simple}

In this section we prove Theorem~\ref{thm:general}. Using our new cycle MSSP data structure, the proof is rather straightforward. 

{
\renewcommand{\themythm}{\ref{thm:general}}
\begin{mythm}
\tradeoffthm
\end{mythm}
}

\begin{proof}
Let $r:=(n^2\lg\lg n) / S$. Note that $r\in[\lg\lg n,n]$ for any $S\in[n\lg\lg n,n^2]$. \\

\paragraph{Preprocessing.} We start by computing an $r$--division. Each
piece has $O(r)$ nodes and $O(\sqrt r)$ boundary nodes incident to a
constant number of holes. 
For each piece $P$ we compute the following:
\begin{enumerate}
\item 
A distance oracle as in Theorem~\ref{thm:linearspace} with $\epsilon=1/\lg r$. 
This takes $O(r\lg r\lg\lg r)$ time and $O(r\lg\lg r)$ space. 
\item For each hole of $P$ (bounded by a cycle in $G$) we compute 
our new cycle MSSP data structure.\footnote{Note that for  $S=o(n\lg n)$ we cannot even afford to store 
Klein's MSSP data
  structure~\cite{conf/soda/Klein05}.}
  Since the number of holes per piece is constant, this
  requires requires $O(n \lg^3 n)$ time and $O(n \lg \lg n)$ space per
  piece.
\label{mssp}
\end{enumerate}
Summing over all pieces, the preprocessing time is
$O(S\lg^3n/\lg\lg n)$ and the  space needed is $O(S)$. 

\paragraph{Query.} Given a pair of nodes $(s,t)$, we compute a shortest
$s$--$t$ path as follows. 
Assume first that $s$ and $t$ are in different pieces.
Let $P$ denote the piece that contains $s$ and let $\partial P$ denote
its boundary.
We compute the distances in $G$ from $\partial P$ to $t$ using 
the cycle MSSP data structures. 
These distances can be obtained in time 
$O(\abs{\partial P}\lg^2 n\lg\lg n)=O(\sqrt{r}\lg^2 n\lg\lg n)$. 
Analogously, we compute the distances in $G$ from $s$ to $\partial P$. 
It remains to find the node $p\in\partial P$ that minimizes $d_G(s,p)+d_G(p,t)$, 
which can be done in $O(\abs{\partial P}) = O(\sqrt r)$ time using a simple sequential search. 

If $s$ and $t$ lie in the same piece,  the shortest
$s$--$t$ path might not visit $\partial P$. To account for this, we query the distance oracle for $P$, which takes
$O(\sqrt{r} \lg^2 r\lg\lg r)$ time. We return the minimum distance found. \qed
\end{proof}

\paragraph{$k$--many distances}
As a consequence, we also obtain an improved algorithm for $k$--many
distances, whenever $k = \Omega(\sqrt{n}/\lg\lg n)$. 
\begin{proof}[of Theorem~\ref{thm:kmanydist}]
For some value of $r$ to be specified below, we preprocess $G$ in time 
$O((n^2/r) \lg^3 n)$, and then we answer each of the $k$ queries in time 
$O(\sqrt{r} \lg^2 r\lg\lg r)$. The total time is
$O((n^2/r) \lg^3 n  + k \sqrt{r} \lg^2 r\lg\lg r)$. This is minimized by setting 
$r = n^{4/3}k^{-2/3} (\lg n/\lg\lg n)^{2/3}$. Note that $r=O(n)$ 
since $k = \Omega(\sqrt{n}\lg n/\lg\lg n)$.
The total running time is thus $O((kn)^{2/3} (\lg n)^{7/3}(\lg\lg n)^{2/3})$.
\qed
\end{proof}

\comment{

our general tradeoff: 
preprocessing : 
(n^2/r) \lg^3 n 
query time :
\sqrt{r} \lg^2 r\lg\lg r

want to balance the time, so 

(n^2/r) \lg^3 n  = k \sqrt{r} \lg^2 r\lg\lg r

assume \lg r = \Theta(\lg n)
we need to set r to 

 r = n^{4/3}k^{-2/3} (\lg n/\lg\lg n)^{2/3} 
 
 then the total time is 
 
 (kn)^{2/3} \lg^{7/3}n(\lg\lg)^{2/3} n 

which we want to be less than Cabello's n^{4/3}\lg^{1/3}n

 (kn)^{2/3} \lg^{7/3}n\lg^{2/3}\lg n < n^{4/3}\lg^{1/3}n
 k < n\lg^{-3}n \lg\lg n

r is required to be at most n 

n >= n^{4/3}k^{-2/3} (\lg n/\lg\lg n)^{2/3} 
k >= n^{1/2} (\lg n/\lg\lg n)

}

\paragraph{Comparison and Discussion.}
The query time of our data structure is at most $O(\sqrt{r} \lg^2 r\lg\lg r)$, which, in terms of $S$, is $O(nS^{-1/2}\lg^{2}n\lg^{3/2}\lg n)$. 
Let us contrast this with Cabello's data structure~\cite{conf/soda/Cabello06} that, for any $S\in[n^{4/3}\lg^{1/3}n,n^2]$ 
has preprocessing time and space $O(S)$ and query time $O(n S^{-1/2}\lg^{3/2}n)$. 
In our construction, we sacrifice a factor of $O(\sqrt{\lg n(\lg\lg n)^3})$ in the query time but we gain a much larger regime for~$S$. 
For the range $S\in[\omega(n\lg n/\lg\lg n), o(n^{4/3}\lg^{1/3}n)]$,
only data structures of size $O(S)$ with query time $O(n^2/S)$  
had been known~\cite{DjidjevWG96} (see also Figure~\ref{fig:tradeoff}).

To conclude this section, let us observe what happens when we gradually decrease the space requirements $S$ from $n^2$ down to $n$, 
and, {\em en passant}, let us pose some open questions. 
For quadratic space (or even slightly below~\cite{WNThesis}), we can obtain constant query time. As soon as we require the space to be $O(n^{2-\epsilon})$ for some $\epsilon>0$, 
the query time increases from constant or poly-logarithmic to a polynomial. It is currently not known whether $\lgO(1)$ is possible or not --- the known lower bounds~\cite{SparseDO,PatrascuRoditty} on the space 
of distance oracles work for non-planar graphs only. 
Further restricting the space, as long as the space available is at least $\Omega(n\lg n)$, the data structure 
can internally {\em store} MSSP data structures. The query time for this regime of $S$ can actually be made slightly faster than what we claim in Theorem~\ref{thm:general} 
(by avoiding the $O(\lg\lg n)$--factor due to the recursion needed in Theorem~\ref{thm:linearspace} and its manifestation in the cycle MSSP data structure). 
The data structure of Nussbaum~\cite{Nussbaum11} obtains a ``clean'' tradeoff with query time proportional to $O(n/\sqrt S)$ for $S\geq n^{4/3}$ ({\em without} logarithmic factors in the query time, currently at the cost of a slower preprocessing algorithm). The obvious open question is whether it is possible to obtain a data structure with space $O(S)$ and query time $O(n/\sqrt S)$ for the whole range of $S\in[n,n^2]$ and without substantial sacrifices with respect to the preprocessing time. Another open question is whether it is possible to improve upon this tradeoff. Note that, for quadratic space, an improvement of an (almost) logarithmic factor is possible~\cite{WNThesis}. 

As soon as we require the space to be $o(n\lg n)$, we cannot afford to store the MSSP data structure anymore and 
we are currently forced to rely on the cycle MSSP data structure (Theorem~\ref{thm:cyclemssp}). The query time increases to the time bound claimed in Theorem~\ref{thm:general}. 
When we further restrict the space to $o(n\lg\lg n)$, say $S=\Theta(n\lg(1/\epsilon))$ for some $\epsilon>0$, the query time increases to $O(n^{1/2+\epsilon})$. A different recursion 
(maybe \`a la~\cite{journals/jcss/HenzingerKRS97}) could potentially reduce this to $\lgO(\sqrt n)$. 

Let us briefly consider the {\em additional space} of the data structure, assuming that storing the graph is free. 
It is known that, for {\em approximate} distances, a query algorithm can run efficiently using a data structure that occupies sublinear additional space~\cite{KKS}. 
For exact distances and sublinear space, nothing better than the linear-time SSSP algorithm~\cite{journals/jcss/HenzingerKRS97} is known.

\section{Distance Oracles with Query Time Quasi-Proportional to the Shortest-Path Length}

We use our new cycle MSSP data structure 
to prove Theorem~\ref{thm:pathlengthquerytime}, which states that there is a distance oracle with 
query time proportional to the shortest-path length. We actually prove two versions, the stronger one 
being a distance oracle with query time proportional to the minimum number of edges (hops) 
on a shortest path. For the stronger version, we need the following assumption, which essentially 
means that approximate shortest paths do not use significantly fewer edges. 

\begin{assumption}
Let $h_G(s,t)$ denote the number of edges on a minimum-hop shortest-path. 
For some constant $\bar\epsilon>0$, all $s$--$t$ paths of length at most $(1+\bar\epsilon)d_G(s,t)$ 
have $\Omega(h_G(s,t))$ edges. 
\label{assumption:hoplength}
\end{assumption}

We restate Theorem~\ref{thm:pathlengthquerytime} and its stronger variant. 

{
\renewcommand{\themythm}{\ref{thm:pathlengthquerytime}}
\begin{mythm}
\pathlengthquerytime

Furthermore, if Assumption~\ref{assumption:hoplength} holds for $G$, the query time is at most 
$O(\min\left\{h\lg^2h\lg\lg h,\sqrt{n}\lg^2n\right\})$ for any pair of nodes $(u,v)$ at {\em hop-distance} $h=h_G(u,v)$. 
\end{mythm}
}


Our main ingredient is a distance oracle for planar graphs with
{\em tree-width}\footnote{See~\cite{halin,RS2,Arn} for definitions and much more on tree-width.}~$w$. 
\begin{theorem}
Let $G$ be a planar graph on $n$ vertices with tree-width~$w$. For any constant $\epsilon>0$ and for any value $S$ in the range 
$S\in[n\lg\lg w,n^2]$, there is a data structure with preprocessing time $O(S\lg^2n+n^{1+\epsilon})$ 
and space $O(S)$ that answers distance queries in $O(\min \set{nS^{-1/2}\lg^{2.5}n, w \lg^2w \lg \lg w})$ time per query. 
\label{thm:treewgeneral}
\end{theorem}
Note that for $S$ superlinear in $n$ but less than roughly $nw$, the oracle cannot make any use of the additional space available and the query algorithm runs in time proportional to $w$ (up to logarithmic factors). Any application should use either space close to linear or more than $\Omega(nw)$.

Using Theorem~\ref{thm:treewgeneral}, the proof of Theorem~\ref{thm:pathlengthquerytime} boils down to a combination of 
{\em slicing} (see for example~\cite{journals/jacm/Baker94,journals/siamcomp/Klein08}), 
{\em local tree-width}~\cite{journals/algorithmica/Eppstein00,journals/algorithmica/DemaineH04}, and  
{\em scaling}.

\begin{proof}[of Theorem~\ref{thm:pathlengthquerytime}]
We repeatedly use Theorem~\ref{thm:treewgeneral} for different subgraphs as follows. 
\paragraph{Slicing and local tree-width.}We use standard techniques~\cite{journals/jacm/Baker94,journals/siamcomp/Klein08} 
and research on the related 
{\em linear local tree-width} property~\cite{journals/algorithmica/Eppstein00,journals/algorithmica/DemaineH04} to 
{\em slice} a planar graph into subgraphs of a certain tree-width: 
\begin{itemize}
\item We compute a breadth-first search tree in the planar dual rooted at an arbitrary face. 
\item The subgraph induced by the nodes at depth $d \in [d_1,d_2]$ has tree-width~$O(d_2-d_1)$~\cite{journals/algorithmica/Eppstein00,journals/algorithmica/DemaineH04}. 
\end{itemize}
If we were only interested in paths using at most $w$ edges, we could {\em (i)} cut 
the graph into slices of depth $w$, and {\em (ii)} check the union of any two consecutive levels containing both endpoints. 
It is straightforward to see that each path on $w$ edges lies completely within one of these unions. 

\paragraph{Scaling.}
We apply the slicing step described in the previous paragraph for different scales. 
\begin{itemize}
\item For every integer $i>0$ with $2^i\leq \sqrt{n}$, we slice the graph into subgraphs $G^i_j$ at depth $r=2^i$; here, $G^i_j$ denotes the graph 
induced by the nodes adjacent to all the faces at depth in $[jr,(j+1)r)$. Every subgraph $G^i_j$ has tree-width at most $O(r)$. 
\item For any two consecutive $G^i_j,G^i_{j+1}$ we compute a distance oracle of size $O(\abs{V(G^i_j\cup G^i_{j+1})}\lg\lg\abs{V(G^i_j\cup G^i_{j+1})})$ with query time $O(r\lg^2r\lg\lg r)$ 
as in Theorem~\ref{thm:treewgeneral}. Since each node is in at most two graphs $G^i_j$ and since each $G^i_j$ participates in at most two distance oracles, the total size of all these distance oracles per level $i$ is $O(n\lg\lg n)$. The total size of our data structure is thus $O(n\lg n\lg\lg n)$. 
\end{itemize}

\paragraph{Which Scale?}
Let the smallest number of edges (or {\em hops}) on any shortest path from $u$ to $v$ be~$h=h_G(u,v)$. 
At query time we use an approximate distance oracle to determine the right scale. 
In the preprocessing algorithm, we also precompute the approximate distance oracle of Thorup~\cite{ThorupJACM04} for $\epsilon=1/2$. This oracle 
can be computed in $O(n\epsilon^{-1}\lg^3n)$ time, it uses space $O(n\epsilon^{-1}\lg n)$, and it answers $(1+\epsilon)$--approximate distance queries in time $O(1/\epsilon)$. 
If Assumption~\ref{assumption:hoplength} holds, 
we instead use that value of $\bar\epsilon$ in the construction of Thorup's distance oracle. The asymptotic space consumption is not increased. 

\paragraph{Query Algorithm.}At query time, given a pair of nodes $(u,v)$ at distance $\ell$, we need to find a level that contains a shortest path. 
We query the approximate distance oracle in time $O(1/\epsilon)$ to obtain an estimate for~$\ell$. Let $\tilde\ell$ denote this estimate. 
We then execute one of the following search algorithms. 

In the case that Assumption~\ref{assumption:hoplength} does not hold, 
we directly query level $i$ for the smallest $i$ with $2^i\geq\tilde\ell$. 
Since all the edge weights are at least $1$, any path of length $\tilde\ell$ has at most $\tilde\ell$ edges and is thus contained in some graph $G_j^i$ at level $i$. 
Since graphs at level $i$ have tree-width $O(2^i)$, and since $2^i=O(\tilde\ell)=O(\ell)$, the running time of the query algorithm is $O(\ell\lg^2\ell\lg\lg\ell)$ as claimed.

If Assumption~\ref{assumption:hoplength} does hold, 
we search the data structure level by level with increasing~$i$ until the first time a distance at most $\tilde\ell$ is found. 
By Assumption~\ref{assumption:hoplength}, we know that any $u$-to-$v$ path of length $\leq(1+\bar\epsilon)\ell$ uses at least $c\cdot h_G(u,v)$ edges for some constant $c\leq1$. 
We therefore search the next $1-\lg_2c$ levels to ensure that we
find a shortest path. The running time is a geometric sum, which is dominated 
by the time to search the last level with tree-width $O(h_G(u,v))$. 
\qed
\end{proof}

\subsection{Distance Oracles for Planar Graphs with Tree-width $o(\sqrt n)$.}
In this section we prove Theorem~\ref{thm:treewgeneral}. 
There exist efficient distance oracles for (not necessarily planar)  graphs with tree-width $w$. 
Chaudhuri and Zaroliagis~\cite{journals/algorithmica/ChaudhuriZ00} provide a distance oracle that uses space $O(w^3n)$ 
and answers distance queries in time $O(w^3\alpha(n))$, where $\alpha(n)$ 
denotes the inverse Ackermann function.
Farzan and Kamali~\cite{conf/icalp/FarzanK11} provide a distance oracle that uses space $O(wn)$ and answers distance queries in time~$O(w^2\lg^3w)$. 

In our application, the tree-width~$w$ may be non-constant up to $O(\sqrt n)$. 
For this reason we cannot use their distance oracles. 
In the following we improve upon their results for the special case where graphs are further assumed to be planar: 
we can obtain space $O(n\lg\lg w)$ and query time $O(w\lg^2w\lg\lg w)$ using our distance oracle. 

In some sense our proof can be seen as an improvement over 
Djidjev's data structure with space $O(S)$ and query time $O(n^2/S)$. His data structure works on any graph with recursive balanced separators of size $O(\sqrt n)$ and, 
furthermore, similar results can be obtained for graphs with separators of general size $f(n)=o(n)$~\cite[Sections~3 and~4]{DjidjevWG96}. 
The data structure with the better trade-off (space $O(S)$ and query time $O(n/\sqrt S)$) however only works for planar graphs~\cite[Section~5]{DjidjevWG96}, 
since it exploits the Jordan Curve Theorem. We can now exploit the Jordan Curve Theorem for planar graphs with smaller separators by observing that, 
for planar graphs with smaller tree-width ($o(\sqrt n)$), 
the size of the Jordan Curve separating the inside from the outside decreases proportionally~\cite{journals/combinatorica/SeymourT94,journals/algorithmica/DornPBF10}. 
These separators are referred to as {\em sphere-cut separators}~\cite{journals/combinatorica/SeymourT94,journals/algorithmica/DornPBF10} and they can be found efficiently. 

\begin{lemma}[{Gu and Tamaki~\cite[Theorem~2]{conf/isaac/GuT09}}]
For any constant $\epsilon>0$ and for any biconnected vertex-weighted planar graph on $n$ nodes with tree-width~$w$, 
there is an $O(n^{1+\epsilon})$--time algorithm that finds a non-self-crossing cycle $C$ of length $O(w/\epsilon)$ 
such that any connected component of $G\setminus C$ has weight at most $3/4$ the total weight. 
\label{lemma:gutamaki}
\end{lemma}
\begin{proof}
The algorithm computes a {\em branch decomposition} as in Gu and Tamaki~\cite[proof of Theorem~2]{conf/isaac/GuT09}. 
It is known that branch-width and tree-width are within constant factors~\cite{journals/jct/RobertsonS91}.

A {\em branch decomposition}~\cite{journals/jct/RobertsonS91} of a graph $G$ is a ternary tree $T$ whose leaves correspond to the edges of $G$. 
Removing any tree edge $e\in T$ creates two connected components $T_1,T_2\subseteq T$, each of which corresponds to a subgraph $G_i$ of $G$, 
induced by the edges corresponding to the set of leaves of $T$ in $T_i$.  
In the branch decomposition tree, we may thus choose the edge $e^*$ that achieves the best balance among all the edges $e\in T$. 

In the algorithm of Gu and Tamaki, each edge of the branch decomposition tree corresponds to a {\em self-non-crossing closed curve} passing through $O(w/\epsilon)$ 
nodes of $G$ (as in sphere-cut decompositions) that encloses $G_1$ and does not enclose $G_2$ (or vice versa). 

Since the degree of each node is at most three, it is possible to choose an edge $e^*$ such that the total weight strictly enclosed by the non-self-crossing cycle $C_{e^*}$ 
and the total weight strictly not enclosed by $C_{e^*}$ are each at most a $3/4$--fraction of the total weight. 
The cycle $C_{e^*}$ is our balanced separator of length $O(w/\epsilon)$. 
\qed
\end{proof}

For an optimal sphere-cut decomposition, the fastest algorithm runs in cubic time~\cite{journals/combinatorica/SeymourT94,journals/talg/GuT08}. For our 
purposes, the constant approximation in Lemma~\ref{lemma:gutamaki} suffices. 

We define a variant of the $r$--division as in Section~\ref{sec:rdiv}, wherein we use either Miller's cycle separator or sphere-cut separators, whichever is smaller.
\def\rwdivision{An {\em $[r,w]$--division} of $G$ with tree-width $w$ is a decomposition into $\Theta(n/r)$
edge-disjoint pieces,
each with $O(r)$ nodes and $O(\min \set{\sqrt{r},w})$ boundary nodes.}
\rwdivision

\def\spherecutrdiv{For any constant $\epsilon>0$ and for any planar graph on $n$ nodes with tree-width~$w$, an {\em $[r,w]$--division} can be
found in time $O(\spherecut\lg n + n \lg r + nr^{-1/2}\lg n)$.}
\begin{lemma}
\spherecutrdiv
\label{lemma:spherecutrdiv}
\end{lemma}
\begin{proof}
Fix $\epsilon$ to be the desired constant $>0$. 
The procedure for obtaining the $[r,w]$--division is the one described in~\cite{nlglgn-mincut-WN10}, but using either Miller's cycle separator or sphere-cut separators, whichever is smaller.
If $\sqrt{r} = O(w)$ then since the separators we use are at least as
small as the separators assumed in the proof of~\cite[Lemmata~2 and~3]{nlglgn-mincut-WN10}, the lemmas apply and the resulting decomposition has $\Theta(n/r)$ pieces, each with $O(r)$ nodes and $O(\sqrt{r})$ boundary nodes.

If $\sqrt{r} = \omega(w)$ then first apply the separator theorem as in the procedure for obtaining a weak $r$--division~\cite[Lemma~2]{nlglgn-mincut-WN10} until every piece has size $O(r)$. Note that since $\sqrt{r} = \omega(w)$, we always use sphere-cut separators, whose size is $O(w)$ (hiding a linear factor of $1/\epsilon$), regardless of the size of the piece we separate. Therefore, by~\cite[Lemma~2]{nlglgn-mincut-WN10}, the number of pieces is $\Theta(n/r)$, and the total number of boundary nodes is $O(n/\sqrt{r})$. However, here we need a tighter bound on the total number of boundary nodes. In the following we show that the total number of boundary nodes is $O(nw/r)$. 

Consider the binary tree whose root corresponds to~$G$, and whose leaves correspond to the pieces of the weak $r$--division. Each internal node in the tree corresponds to a piece in the recursive decomposition that either consists of too many nodes or contains too many holes and is therefore separated into two subpieces using a sphere-cut separator. 
As in the proof of~\cite[Lemma~2]{nlglgn-mincut-WN10}, for any boundary vertex $v$ in the
weak $r$--division, let $b(v)$ denote one less than the number of pieces containing $v$ as a boundary vertex. 
Let $B(n)$ be the sum of $b(v)$ over all such $v$. Every time a
separator is used, at most $cw$ boundary nodes are introduced for some $c>0$, and the piece to which these nodes belong to is split into two pieces, so the number of pieces to which these nodes belong to increases by one. Therefore, $B(n)$ is bounded by the number of internal nodes times $cw$.
Since the tree is binary, the number of internal nodes is bounded by
the number of leaves, so $B(n) = O(nw/r)$, which shows that the total
number of boundary nodes in this weak $r$--division is $O(nw/r)$.

Next, consider the procedure in~\cite[Lemma~3]{nlglgn-mincut-WN10}, which further divides the pieces of the weak $r$--division to make sure that the number of boundary nodes in each piece is small enough. In our case we want to limit the number of boundary nodes per piece to at most $c'w$ for some $c'$ (which we set $>c$ below).  Let $t_i$ denote the number of pieces in the weak $r$--division with exactly $i$ boundary nodes. Note that $\sum_i it_i = \sum_{v\in V_B}(b(v) + 1)$, where $V_B$ is the set of boundary vertices over all
pieces in the weak $r$--division. Hence, $\sum_i it_i  \leq 2B(n)$, so by the bound on $B(n)$, $\sum_i it_i = O(nw/r)$.

In the weak $r$--division, consider a piece $P$ with $i > c'w$ boundary vertices. When the above procedure splits $P$, each of the subpieces contains at most a constant fraction of the boundary vertices of~$P$. Hence, after
$di/(c'w)$ splits of $P$ for some constant $d$, all subpieces contain at most $c'w$ boundary vertices. This results 
in at most $1 + di/(c'w)$ subpieces and at most $cw$ new boundary
vertices per split. We may choose $c'$ to be sufficiently larger than $c$. The total number of new boundary vertices introduced by the above procedure is thus
\[
  \sum_i cw(di/(c'w))t_i\leq d\sum_i it_i = O(nw/r)
\]
and the number of new pieces is at most
\[
  \sum_i (di/(c'w))t_i = O(n/r).
\]
Hence, the procedure generates an $[r,w]$--division. Since a sphere-cut separator can be found in $O(n^{1+\epsilon})$ time, a weak $r$--division can be found in $O(n^{1+\epsilon}\lg n) = O(n^{1+\epsilon'})$ time, and refining it into an $[r,w]$--division can 
be done within this time bound.
\qed
\end{proof}

\begin{proof}[of Theorem~\ref{thm:treewgeneral}.]
The data structure for planar graphs with smaller tree-width is essentially the same as the data structure for general planar graphs 
as described in Section~\ref{sec:general:wholerange}. 
The main difference is that, when computing a cycle separator, instead of Miller's algorithm to find a cycle separator of length $O(\sqrt n)$, 
we use sphere-cut separators and the corresponding $[r,w]$--division as in Lemma~\ref{lemma:spherecutrdiv}. 
\qed
\end{proof}

\section{Conclusion.}

We introduce a new data structure to answer distance queries between any node $v$ and all the 
nodes on a not-too-long cycle of a planar graph. Using this tool, 
we significantly improve the worst-case query times for distance oracles with low space requirements 
 $S$ (down from $O(n^2/S)$ to $\lgO(\sqrt{n^2/S})$). 
Furthermore, for linear space, we improve the query time down from $O(n)$ to $O(n^{1/2+\epsilon})$ using adaptive recursion. 
We also give the first distance oracle that actually exploits the tree-width of a planar graph, particularly if it is $o(\sqrt n)$ 
and, as an application, we give a distance oracle whose query time is roughly proportional to the shortest-path 
length. 
Similar behavior of practical methods had been observed experimentally before but could not be proven until the current work. 

An interesting and important open question is whether there is another tradeoff curve below the space $O(S)$ and query time $O(n/\sqrt S)$ curve.

\section*{Acknowledgments.}
We thank Hisao Tamaki for helpful discussions on~\cite{conf/isaac/GuT09}.  

\bibliographystyle{alpha}
\bibliography{planaradq}

\newcommand{\etalchar}[1]{$^{#1}$}
\begin{thebibliography}{BDDW09}

\bibitem[ACC{\etalchar{+}}96]{ArikatiCCDSZ96}
Srinivasa~Rao Arikati, Danny~Z. Chen, L.~Paul Chew, Gautam Das, Michiel H.~M.
  Smid, and Christos~D. Zaroliagis.
\newblock Planar spanners and approximate shortest path queries among obstacles
  in the plane.
\newblock In {\em 4th European Symposium on Algorithms (ESA)}, pages 514--528,
  1996.

\bibitem[ADF{\etalchar{+}}11]{conf/icalp/AbrahamDFGW11}
Ittai Abraham, Daniel Delling, Amos Fiat, Andrew~V. Goldberg, and Renato
  Fonseca~F. Werneck.
\newblock {VC}-dimension and shortest path algorithms.
\newblock In {\em 38th International Colloquium on Automata, Languages, and
  Programming (ICALP)}, pages 690--699, 2011.

\bibitem[AFGW10]{HighwayDimension}
Ittai Abraham, Amos Fiat, Andrew~V. Goldberg, and Renato Fonseca~F. Werneck.
\newblock Highway dimension, shortest paths, and provably efficient algorithms.
\newblock In {\em 21st ACM-SIAM Symposium on Discrete Algorithms (SODA)}, 2010.

\bibitem[AP89]{Arn}
Stefan Arnborg and Andrzej Proskurowski.
\newblock Linear time algorithms for {NP}--hard problems restricted to partial
  $k$--trees.
\newblock {\em Discrete Applied Mathematics}, 23(1):11--24, 1989.

\bibitem[AT05]{conf/isaac/ArgeT05}
Lars Arge and Laura Toma.
\newblock External data structures for shortest path queries on planar
  digraphs.
\newblock In {\em 16th International Symposium on Algorithms and Computation
  (ISAAC)}, pages 328--338, 2005.

\bibitem[Bak94]{journals/jacm/Baker94}
Brenda~S. Baker.
\newblock Approximation algorithms for {NP}--complete problems on planar
  graphs.
\newblock {\em Journal of the ACM}, 41(1):153--180, 1994.
\newblock Announced at FOCS 1983.

\bibitem[BCK{\etalchar{+}}10]{conf/ciac/BauerCKKW10}
Reinhard Bauer, Tobias Columbus, Bastian Katz, Marcus Krug, and Dorothea
  Wagner.
\newblock Preprocessing speed-up techniques is hard.
\newblock In {\em 7th International Conference on Algorithms and Complexity
  (CIAC)}, pages 359--370, 2010.

\bibitem[BDDW09]{ShortcutProblem}
Reinhard Bauer, Gianlorenzo D'Angelo, Daniel Delling, and Dorothea Wagner.
\newblock The shortcut problem - complexity and approximation.
\newblock In {\em 35th Conference on Current Trends in Theory and Practice of
  Computer Science (SOFSEM)}, pages 105--116, 2009.

\bibitem[BFSS07]{TransitNodesScience}
Holger Bast, Stefan Funke, Peter Sanders, and Dominik Schultes.
\newblock Fast routing in road networks with transit nodes.
\newblock {\em Science}, 316(5824):566, 2007.

\bibitem[BSWN10]{conf/focs/BorradaileSW10}
Glencora Borradaile, Piotr Sankowski, and Christian Wulff-Nilsen.
\newblock Min $st$-cut oracle for planar graphs with near-linear preprocessing
  time.
\newblock In {\em 51st IEEE Symposium on Foundations of Computer Science
  (FOCS)}, pages 601--610, 2010.

\bibitem[Cab06]{conf/soda/Cabello06}
Sergio Cabello.
\newblock Many distances in planar graphs.
\newblock In {\em 17th ACM-SIAM Symposium on Discrete Algorithms (SODA)}, pages
  1213--1220, 2006.
\newblock A preprint of the journal version is available in the University of
  Ljubljana preprint series, Vol. 47 (2009), 1089.

\bibitem[CX00]{stoc/ChenX00}
Danny~Ziyi Chen and Jinhui Xu.
\newblock Shortest path queries in planar graphs.
\newblock In {\em 32nd ACM Symposium on Theory of Computing (STOC)}, pages
  469--478, 2000.

\bibitem[CZ00]{journals/algorithmica/ChaudhuriZ00}
Shiva Chaudhuri and Christos~D. Zaroliagis.
\newblock Shortest paths in digraphs of small treewidth. part {I}: Sequential
  algorithms.
\newblock {\em Algorithmica}, 27(3):212--226, 2000.
\newblock Announced at ICALP 1995.

\bibitem[DH04]{journals/algorithmica/DemaineH04}
Erik~D. Demaine and Mohammad~Taghi Hajiaghayi.
\newblock Diameter and treewidth in minor-closed graph families, revisited.
\newblock {\em Algorithmica}, 40(3):211--215, 2004.

\bibitem[Dij59]{Dijkstra59}
Edsger~Wybe Dijkstra.
\newblock A note on two problems in connexion with graphs.
\newblock {\em Numerische Mathematik}, 1:269--271, 1959.

\bibitem[Dji96]{DjidjevWG96}
Hristo Djidjev.
\newblock Efficient algorithms for shortest path problems on planar digraphs.
\newblock In {\em 22nd International Workshop on Graph-Theoretic Concepts in
  Computer Science (WG)}, pages 151--165, 1996.

\bibitem[DPBF10]{journals/algorithmica/DornPBF10}
Frederic Dorn, Eelko Penninkx, Hans~L. Bodlaender, and Fedor~V. Fomin.
\newblock Efficient exact algorithms on planar graphs: Exploiting sphere cut
  decompositions.
\newblock {\em Algorithmica}, 58(3):790--810, 2010.
\newblock Announced at ESA 2005.

\bibitem[DPZ00]{journals/algorithmica/DjidjevPZ00}
Hristo Djidjev, Grammati~E. Pantziou, and Christos~D. Zaroliagis.
\newblock Improved algorithms for dynamic shortest paths.
\newblock {\em Algorithmica}, 28(4):367--389, 2000.

\bibitem[EG08]{conf/gis/EppsteinG08}
David Eppstein and Michael~T. Goodrich.
\newblock Studying (non-planar) road networks through an algorithmic lens.
\newblock In {\em 16th ACM SIGSPATIAL International Symposium on Advances in
  Geographic Information Systems (ACM-GIS)}, page~16, 2008.

\bibitem[Epp99]{journals/jgaa/Eppstein99}
David Eppstein.
\newblock Subgraph isomorphism in planar graphs and related problems.
\newblock {\em Journal of Graph Algorithms and Applications}, 3(3), 1999.
\newblock Announced at SODA 1995.

\bibitem[Epp00]{journals/algorithmica/Eppstein00}
David Eppstein.
\newblock Diameter and treewidth in minor-closed graph families.
\newblock {\em Algorithmica}, 27(3-4):275--291, 2000.

\bibitem[FK11]{conf/icalp/FarzanK11}
Arash Farzan and Shahin Kamali.
\newblock Compact navigation and distance oracles for graphs with small
  treewidth.
\newblock In {\em 38th International Colloquium on Automata, Languages, and
  Programming (ICALP)}, pages 268--280, 2011.

\bibitem[FMS91]{conf/wg/FeuersteinM91}
Esteban Feuerstein and Alberto Marchetti-Spaccamela.
\newblock Dynamic algorithms for shortest paths in planar graphs.
\newblock In {\em 17th International Workshop on Graph-Theoretic Concepts in
  Computer Science (WG)}, pages 187--197, 1991.

\bibitem[FR06]{journals/jcss/FakcharoenpholR06}
Jittat Fakcharoenphol and Satish Rao.
\newblock Planar graphs, negative weight edges, shortest paths, and near linear
  time.
\newblock {\em Journal of Computer and System Sciences}, 72(5):868--889, 2006.
\newblock Announced at FOCS 2001.

\bibitem[Fre87]{journals/siamcomp/Frederickson87}
Greg~N. Frederickson.
\newblock Fast algorithms for shortest paths in planar graphs, with
  applications.
\newblock {\em SIAM Journal on Computing}, 16(6):1004--1022, 1987.

\bibitem[GSSD08]{GeisbergerSSD08}
Robert Geisberger, Peter Sanders, Dominik Schultes, and Daniel Delling.
\newblock Contraction hierarchies: Faster and simpler hierarchical routing in
  road networks.
\newblock In {\em 7th International Workshop on Experimental Algorithms (WEA)},
  pages 319--333, 2008.

\bibitem[GT08]{journals/talg/GuT08}
Qian-Ping Gu and Hisao Tamaki.
\newblock Optimal branch-decomposition of planar graphs in {$O(n^3)$} time.
\newblock {\em ACM Transactions on Algorithms}, 4(3), 2008.
\newblock Announced at ICALP 2005.

\bibitem[GT09]{conf/isaac/GuT09}
Qian-Ping Gu and Hisao Tamaki.
\newblock Constant-factor approximations of branch-decomposition and largest
  grid minor of planar graphs in {$O(n^{1+\epsilon})$} time.
\newblock In {\em 20th International Symposium on Algorithms and Computation
  (ISAAC)}, pages 984--993, 2009.

\bibitem[GW05]{conf/alenex/GoldbergW05}
Andrew~V. Goldberg and Renato Fonseca~F. Werneck.
\newblock Computing point-to-point shortest paths from external memory.
\newblock In {\em 7th Workshop on Algorithm Engineering and Experiments
  (ALENEX)}, pages 26--40, 2005.

\bibitem[Hal76]{halin}
Rudolf Halin.
\newblock {$S$}-functions for graphs.
\newblock {\em Journal of Geometry}, 8(1-2):171--186, 1976.

\bibitem[HKMS09]{ArcflagsDimacs}
Moritz Hilger, Ekkehard K{\"o}hler, Rolf~H. M{\"o}hring, and Heiko Schilling.
\newblock Fast point-to-point shortest path computations with arc-flags.
\newblock {\em DIMACS Series in Discrete Mathematics and Theoretical Computer
  Science}, 74:41--72, 2009.
\newblock Papers of the 9th DIMACS Implementation Challenge: Shortest Paths,
  2006.

\bibitem[HKRS97]{journals/jcss/HenzingerKRS97}
Monika~Rauch Henzinger, Philip~Nathan Klein, Satish Rao, and Sairam
  Subramanian.
\newblock Faster shortest-path algorithms for planar graphs.
\newblock {\em Journal of Computer and System Sciences}, 55(1):3--23, 1997.
\newblock Announced at STOC 1994.

\bibitem[HMZ03]{journals/dam/HutchinsonMZ03}
David~A. Hutchinson, Anil Maheshwari, and Norbert Zeh.
\newblock An external memory data structure for shortest path queries.
\newblock {\em Discrete Applied Mathematics}, 126(1):55--82, 2003.
\newblock Announced at COCOON 1999.

\bibitem[JHR96]{conf/cikm/JingHR96}
Ning Jing, Yun-Wu Huang, and Elke~A. Rundensteiner.
\newblock Hierarchical optimization of optimal path finding for transportation
  applications.
\newblock In {\em 5th International Conference on Information and Knowledge
  Management (CIKM)}, pages 261--268, 1996.

\bibitem[Joh77]{Johnson77}
Donald~B. Johnson.
\newblock Efficient algorithms for shortest paths in sparse networks.
\newblock {\em Journal of the ACM}, 24(1):1--13, 1977.

\bibitem[KK06]{journals/talg/KowalikK06}
Lukasz Kowalik and Maciej Kurowski.
\newblock Oracles for bounded-length shortest paths in planar graphs.
\newblock {\em ACM Transactions on Algorithms}, 2(3):335--363, 2006.
\newblock Announced at STOC 2003.

\bibitem[KKS11]{KKS}
{Ken-ichi} Kawarabayashi, Philip~Nathan Klein, and Christian Sommer.
\newblock Linear-space approximate distance oracles for planar, bounded-genus,
  and minor-free graphs.
\newblock In {\em 38th International Colloquium on Automata, Languages, and
  Programming (ICALP)}, pages 135--146, 2011.

\bibitem[Kle02]{conf/soda/Klein02}
Philip~Nathan Klein.
\newblock Preprocessing an undirected planar network to enable fast approximate
  distance queries.
\newblock In {\em 13th ACM-SIAM Symposium on Discrete Algorithms (SODA)}, pages
  820--827, 2002.

\bibitem[Kle05]{conf/soda/Klein05}
Philip~Nathan Klein.
\newblock Multiple-source shortest paths in planar graphs.
\newblock In {\em 16th ACM-SIAM Symposium on Discrete Algorithms (SODA)}, pages
  146--155, 2005.

\bibitem[Kle08]{journals/siamcomp/Klein08}
Philip~Nathan Klein.
\newblock A linear-time approximation scheme for {TSP} in undirected planar
  graphs with edge-weights.
\newblock {\em SIAM Journal on Computing}, 37(6):1926--1952, 2008.
\newblock Announced at FOCS 2005.

\bibitem[KMS05]{conf/wea/KohlerMS05}
Ekkehard K{\"o}hler, Rolf~H. M{\"o}hring, and Heiko Schilling.
\newblock Acceleration of shortest path and constrained shortest path
  computation.
\newblock In {\em 4th International Workshop on Experimental and Efficient
  Algorithms (WEA)}, pages 126--138, 2005.

\bibitem[KMW10]{journals/talg/KleinMW10}
Philip~Nathan Klein, Shay Mozes, and Oren Weimann.
\newblock Shortest paths in directed planar graphs with negative lengths: A
  linear-space {$O(n\log^2n)$}-time algorithm.
\newblock {\em ACM Transactions on Algorithms}, 6(2), 2010.
\newblock Announced at SODA 2009.

\bibitem[Lau04]{lauther04}
Ulrich Lauther.
\newblock An extremely fast, exact algorithm for finding shortest paths in
  static networks with geographical background.
\newblock In {\em Geoinformation und Mobilit{\"a}t --- von der Forschung zur
  praktischen Anwendung}, volume~22, pages 219--230, 2004.

\bibitem[Mil86]{journals/jcss/Miller86}
Gary~L. Miller.
\newblock Finding small simple cycle separators for 2--connected planar graphs.
\newblock {\em Journal of Computer and System Sciences}, 32(3):265--279, 1986.
\newblock Announced at STOC 1984.

\bibitem[MWN10]{esa/MozesW10}
Shay Mozes and Christian Wulff-Nilsen.
\newblock Shortest paths in planar graphs with real lengths in
  {$O(n\log^2n/\log\log n)$} time.
\newblock In {\em 18th European Symposium on Algorithms (ESA)}, pages 206--217,
  2010.

\bibitem[Nus11]{Nussbaum11}
Yahav Nussbaum.
\newblock Improved distance queries in planar graphs.
\newblock In {\em 12th International Symposium on Algorithms and Data
  Structures (WADS)}, pages 642--653, 2011.

\bibitem[PR10]{PatrascuRoditty}
Mihai Patrascu and Liam Roditty.
\newblock Distance oracles beyond the {Thorup--Zwick} bound.
\newblock In {\em 51st IEEE Symposium on Foundations of Computer Science
  (FOCS)}, 2010.

\bibitem[RS86]{RS2}
Neil Robertson and Paul~D. Seymour.
\newblock Graph minors. {II}. algorithmic aspects of tree-width.
\newblock {\em Journal of Algorithms}, 7:309--322, 1986.

\bibitem[RS91]{journals/jct/RobertsonS91}
Neil Robertson and Paul~D. Seymour.
\newblock Graph minors. {X.} obstructions to tree-decomposition.
\newblock {\em J. Comb. Theory, Ser. B}, 52(2):153--190, 1991.

\bibitem[Sch98]{journals/siamcomp/Schmidt98}
Jeanette~P. Schmidt.
\newblock All highest scoring paths in weighted grid graphs and their
  application to finding all approximate repeats in strings.
\newblock {\em SIAM Journal on Computing}, 27(4):972--992, 1998.
\newblock Announced at ISTCS 1995.

\bibitem[Som10]{SommerThesis}
Christian Sommer.
\newblock {\em Approximate Shortest Path and Distance Queries in Networks}.
\newblock PhD thesis, The University of Tokyo, 2010.

\bibitem[Som11]{journals/corr/abs-1109-2641}
Christian Sommer.
\newblock More compact oracles for approximate distances in planar graphs.
\newblock {\em CoRR}, abs/1109.2641, 2011.

\bibitem[SS05]{SchultesSanders05}
Peter Sanders and Dominik Schultes.
\newblock Highway hierarchies hasten exact shortest path queries.
\newblock In {\em 13th Annual European Symposium on Algorithms (ESA)}, pages
  568--579, 2005.

\bibitem[ST94]{journals/combinatorica/SeymourT94}
Paul~D. Seymour and Robin Thomas.
\newblock Call routing and the ratcatcher.
\newblock {\em Combinatorica}, 14(2):217--241, 1994.

\bibitem[SVY09]{SparseDO}
Christian Sommer, Elad Verbin, and Wei Yu.
\newblock Distance oracles for sparse graphs.
\newblock In {\em 50th IEEE Symposium on Foundations of Computer Science
  (FOCS)}, pages 703--712, 2009.

\bibitem[Tho04]{ThorupJACM04}
Mikkel Thorup.
\newblock Compact oracles for reachability and approximate distances in planar
  digraphs.
\newblock {\em Journal of the ACM}, 51(6):993--1024, 2004.
\newblock Announced at FOCS 2001.

\bibitem[WN10a]{WNThesis}
Christian Wulff-Nilsen.
\newblock {\em Algorithms for Planar Graphs and Graphs in Metric Spaces}.
\newblock PhD thesis, University of Copenhagen, 2010.

\bibitem[WN10b]{nlglgn-mincut-WN10}
Christian Wulff-Nilsen.
\newblock Min $st$-cut of a planar graph in {$O(n \log\log n)$} time.
\newblock {\em CoRR}, abs/1007.3609, 2010.

\bibitem[Zar08]{reference/algo/Zaroliagis}
Christos Zaroliagis.
\newblock Engineering algorithms for large network applications.
\newblock In {\em Encyclopedia of Algorithms}. 2008.

\end{thebibliography}

\end{document}